\documentclass[]{interact}

\usepackage{epstopdf}
\usepackage[caption=false]{subfig}
\usepackage{apacite}
\usepackage[longnamesfirst,sort]{natbib}
\bibpunct[, ]{(}{)}{;}{a}{,}{,}

\usepackage{amssymb}
\usepackage{amsmath,amssymb,amsfonts}
\usepackage{algorithmic}
\usepackage{amsthm}
\usepackage{bm}
\usepackage{algorithm,float}

\newcommand{\Rmnum}[1]{\expandafter\@slowromancap\romannumeral #1@}

\theoremstyle{plain}
\newtheorem{theorem}{Theorem}[section]
\newtheorem{lemma}[theorem]{Lemma}

\theoremstyle{definition}
\newtheorem{definition}[theorem]{Definition}

\theoremstyle{remark}

\newtheorem{notation}{Notation}

\begin{document}
\title{Robust self-triggered DMPC for linear discrete-time systems with local and global constraints}

\author{
\name{Zhengcai Li\textsuperscript{a}\thanks{CONTACT Zhengcai Li. Email: lizhengcai0224@njust.edu.cn}, Changbing Ma\textsuperscript{a} and Huiling Xu\textsuperscript{a}}
\affil{\textsuperscript{a}School of Science, Nanjing University of Science and Technology, Nanjing, China}}

\maketitle

\begin{abstract}
This paper proposes a robust self-triggered distributed model predictive control (DMPC) scheme for a family of Discrete-Time linear systems with local (uncoupled) and global (coupled) constraints. To handle the additive disturbance, tube-based method is proposed for the satisfaction of local state and control constraints. Meanwhile, A special form of constraints tightening is given to guranteen the global coupled constraints. The self-triggering mechanism help reduce the computation burden by skip insignificant iteration steps, which determine a certain sampling instants to solve the DMPC optimization problem in parallel ways. The DMPC optimization problem is construct as a dual form, and sloved distributedly based on the Alternative Direction Multiplier Method (ADMM) with some known simplifications. Recursive feasibility and input-to-state stability of the closed-loop system are shown, the performance of proposed scheme is demonstrated by a simulation example.
\end{abstract}


\section{Introduction}

This paper considers the Robust Distributed Model Predictive Control (RDMPC) of M discrete-time linear dynamical systems  with disturbance, which formulate as
\begin{align}
x^i(t+1)&=A^ix^i(t)+B^iu^i(t)+w^i(t)\\
x^i(t)\in \mathcal{X}^i,& u^i(t)\in \mathcal{U}^i, w^i(t)\in \mathcal{W}^i, i=1,\dots,M
\end{align}
With all of them have to respect a global coupled constraints as following
\begin{align}
\sum_{i=1}^{M}(\Psi^i_xx^i(t)+\Psi^i_uu^i(t))\leq\textbf{1}_p, \text{for all}\quad t
\end{align} 
where $x^i\in\mathbb{R}^{n_i}$, $u^i\in\mathbb{R}^{m_i}$ and $w^i\in\mathbb{R}^{n_i}$ are the state, input and disturbance of of the $i^{th}$ system; $\mathcal{X}^i$ and $\mathcal{U}^i$ are the corresponding state and input constraints for the $i^{th}$ system;  $\mathcal{W}^i\triangleq\{w^i\in R^{n_i} | \parallel w^i \parallel \le\bar{w^i}\}$ is a compact set containing the origin in its interior; the matrices $\Psi^i_x\in\mathbb{R}^{p\times n_i}$ and $\Psi^i_u\in\mathbb{R}^{p\times m_i}$ denote the coupled constraints of all $M$  systems.\\
\indent Model predictive control has been proved to be highly successful in comparison with alternative methods of multivariable control due to its independent adoption by the process industries \citep{MAYNE20142967}. The study of RMPC is an active area of research since the robust stability has been studied using the theory of input-to-state stability (ISS)\citep{JIANG2001857}\citep{Sontag199520}. Some examples of the use of ISS theory include\citep{Limon2009}\citep{5129691} The classical Lyapunov and input-to-state definitions of stability have subtle differences especially when the system is discontinuous; these related issues including implications for MPC are thoroughly explored in\citep{6376100}.\\
\indent The useful concept of tube-based method originated with the seminal papers in 1970s and remains an active area of research\citep{Kurzhanski1993}. It mainly construct nominal MPC system ignoring the additive distrubance to simplify the optimization problem in predictive step. It was initially used in MPC \citep{CHISCI20011019} for linear system, and the tube-based procedure has also been extended to control nonlinear systems \citep{doi:10.1002/rnc.1758}. Aperiodic MPC is one another approach to lose the stress of computation burden, there are two main types of aperiodic MPC, which are event-triggered MPC\citep{ShuaiLiu2017} \citep{7922495} and self-triggered MPC\citep{BERGLIND2012342} \citep{GOMMANS20141279}\citep{6862397}. The detail introduction about the trigering mechanism have been shown in\citep{6425820}. Recently,\citep{DAI20191446} have proposed a novel self-triggered MPC algorithm for linear system subject to both disturbances and state/input constraints, and an self-triggered MPC algorithm with adaptive predictive horizon for perturbed nonlinear system\citep{8667353} have been proposed further. These latest literature shows that tube-based method and self-triggered mechanism can be combined perfectly in discrete-time linear system with bounded disturbance.\\
\indent There exist many complex systems require control, but the traditional centralized control are unsuitable because of their complexity. The complexity gives rise to modeling and data collection issues, raises computational and communication problems, and make centralized control impractical. A very large literature on distributed control and distributed model predictive control (DMPC) has emerged for this reason\citep{WANG20102053}. One popular area in DMPC are state/input globle coupled when the constraints come to $(3)$. \citep{doi:10.1080/00207170701491070}gives a method of sequential process, it arrange every subsystem updated in a certain sequnce until all system complete their iterations. Another approach called cooperative MPC method given by\citep{Trodden2014}\citep{TRODDEN201498}. The main idea of that approach is that all systems within a cooperating set are optimized jointly while systems outside set following their old predictive controls.\\
\indent A reasonable appoach for system $(1)-(3)$ without disturbance is proposed in\citet{bertsekas1998nonlinear}. The method can achieve overall optimality through dual problem involving the Lagrange function, and the dual variable is function as consensus variable in distributed optimal control problem. Until the Alternative Direction Method of Multipliers(ADMM) is reproposed by\citep{10.1145/2020408.2020410}, and its nice numerical performance in distributed system in many appliations motivated some scholars restudy the distributed Optimal control problem(DOCP) using distributed ADMM recently\citep{FARINA20121088}\citep{7799069}. While the computation expendiency can be very high according to huge iteration steps in ADMM, a brand new method has been proposed to handle premature termination ADMM algorithm through constraints tightening\citep{WANG2017184}.\\
\indent Inspired by the above work of distributed Model Predictive Control under distributed ADMM. A novel robust self-triggered distributed model pridictive control is proposed for a family of discrete-time linear perturbed systems with local and global constraints in this paper. The main contributions of this paper organized as follows.\\
\indent 1) It mainly contributes the tube-based method in robust MPC to distributed discrete-time linear system. The distribution of system gives rise in terminal set constructure, constraints satisfaction and stability issues, especially the perturbed distributed system makes global constraints $(3)$ tough to satisfied. A special form of constraints tightening is established so that  all the above difficulties are solved.\\
\indent 2) Sufficient condition about local and global constraints of bounded disturbance is given to gurantee recursive feasibility in Optimal control problem(OCP), which can link the tube-based method in RMPC and constraints tightening in distributed ADMM properly.\\
\indent 3) A parallel DMPC algorithm based on augmented lagrange method(ALM) is given in this paper,  which applies a self-triggering mechanism in RDMPC to lessen computation burdan of optimization control problem as well. It shows the advantage of partial parallel algorithm in iteration times efficiency than traditional central methods.\\
\indent The paper is organised as follows. In Section 2, the method of tube-based to handle bounded dissturbance is given and some useful properties under that approach is introduced. In Section 3, a robust distributed model predictive control scheme is designed and a robust self-triggered DMPC algorithm is proposed. Main results involving recursive feasibility and stability are developed in Section 4. Section 5 gives a numerical example to illustrating the results obtained and conclusion of the paper is given in Section 6. 
\begin{notation}
	The set of $n^i$-dimensional Euclidean space, $m^i$-dimensional Euclidean space, the set of $p \times n^i$ real matrices and $p \times m^i$ real matrices are denoted $\mathbb{R}^{n_i},\mathbb{R}^{m_i},\mathbb{R}^{p\times n_i}$, and $\mathbb{R}^{p\times m_i}$ respectively. $\textbf{1}_p$ denote p-dimensional Column vector whose element is full of 1. The notations $N_{[0,N-1]}, \mathbb{Z}_M$ indicate the sets $\{r\in \mathbb{N}|0\le r\le N-1\}$ and $\{r\in \mathbb{Z}|r\le M\}$ respectively. for $W \in \mathbb{R}^{n\times n}$, the notation $W>0$ indicate $W$ is positive definite, $\parallel W \parallel$ denotes the 2-norm of $W$ and $\mathbb{S}_{++}^{n}$ denotes the space of symmetric $n\times n$ positive definite matrices. for $W \in\mathbb{S}_{++}^{n}$ and $x\in \mathbb{R}^n$,$\parallel x\parallel^2 = x^Tx$ and $ \parallel x \parallel_W^2 = x^TWx$. Given $\mathcal{X},\mathcal{Z}\subseteq\mathbb{R}^n$ and $A\in \mathbb{R}^{n\times n}$, $A\mathcal{X} = \{Ax|x\in\mathcal{X}\},\mathcal{X} \oplus \mathcal{Z} = \{x+z| x\in\mathcal{X},z\in\mathcal{Z}\},\mathcal{X} \ominus \mathcal{Z} = \{x+z| x\in\mathcal{X},z\in\mathcal{Z}\}$. We denote the smallest and largest eigenvalues of a matrix as $\lambda_{max}(\centerdot),\lambda_{min}(\centerdot)$. $\mathcal{V} = \{1,2,\dots,M\}$ and $ \mathcal{E} \subset \mathcal{V} \times \mathcal{V}$ are the vertex set and edge set of undirect graph $G$. The $i$-step-ahead predicted value of a variable at time $t_k$ denote as $(t_k+i|t_k)$.
\end{notation}
For the sake of readability, some proofs are located in Appendix.
\section{Preliminaries and properties of tube-based method}
This section reviews some results in tube-based RMPC and other related concepts. For a expository reason, we using tube-based approach to achieve the control goal. We can take full advantage of the optimal control input sequence calculated at each sampling time instead of just utilising the first element of the sequence. Inter-sampling time is chosen by a self-triggered condition, which actually through a minimun function. 
We first define the nominal system with a given prediction horizon $N\in\mathbb{N}$ by neglecting the disturbance in$(1)$
\begin{equation}
z^i(k+l+1|k)=A^iz^i(k+l|k)+B^iu^i(k+l|k),  l\in N_{[0,N-1]}
\end{equation}
where $z^i(k+l|k)$ is the predicted nominal state at step $k$, and we denotes the error between the actual state and the nominal one as $e^i(k+l|k)\triangleq x^i(k+l|k)-z^i(k+l|k)$, with initial conditions $z^i(k|k)=x^i(k|k)$ and $e^i(k|k)=0$. It evolve as
\begin{align}
e^i(k+l|k)=&A^ie^i(k+l-1|k)+w^i(k+l-1)\\
=&{A^i}^{l-1}w^i(k)+{A^i}^{l-2}w^i(k+1)+ \dots 
+w^i(k+l-1)\notag
\end{align}
\indent If we can assure the boundness of above error $e_i(k+l|k)$ which we'll prove in section 4, and ignoring the based coupled constraints $(3)$, then we can construct the OCP at state $x^i(t_k)$ as
\begin{subequations}
	\begin{align}
	\mathbb{P}(x^i(t_k)) \triangleq &\min_{u^i(t_k)}J^i(x^i(t_k),u^i(t_k))\notag\\
	=&\sum_{l=0}^{N-1}[\parallel z^i(t_k+l|t_k)\parallel_{Q^i}^2+\parallel u^i(t_k+l|t_k)\parallel_{R^i}^2]+\parallel z^i(t_k+N|t_k)\parallel_{P^i}^2
	\end{align}
	subject to
	\begin{align}
	&z^i(t_k|t_k)=x_i(t_k)\\
	&z^i(t_k+N|t_k)\in 	T^i_f\\
	&\forall l\in\mathbb{N}_{[0,N-1]}:\notag\\
	&z^i(t_k+l+1|t_k)=A^iz^i(t_k+l|t_k)+B^iu^i(k+l|k)\\
	&z^i(t_k+l|t_k)\in \mathcal{Z}_l^i\\
	&u^i(t_k+l|t_k)\in \mathcal{U}^i
	\end{align}
\end{subequations}
\indent The nominal predictive state constraints set is defined as $\mathcal{Z}_l^i\triangleq\mathcal{X}^i\ominus\mathcal{R}_l^i,l\in \mathbb{N}_{[1,N-1]}$, with $\mathcal{R}_l^i\triangleq\oplus_{j=0}^{l-1}{A^i}^j\mathcal{W}^i$. $Q^i>0,R^i>0$ and $P^i>0$ are the weighting matrices of cost function. 
Regarding to the terminal set $\mathcal{X}^i_\varepsilon$, we first construct a set $\mathcal{X}^i_r \triangleq \{z^i \in R^{n_i}| \parallel z^i \parallel_{P^i} \le r^i\}$ for this nominal system$(5)$, such that the set $\mathcal{X}^i_r$ is maximal Robust Control Invariant Set(RCIS) under local state feedback control law : $u^i(z^i) = K^iz^i$, which also need to satisfied with following basic constraints:
\begin{align}
&\forall z^i \in \mathcal{X}^i_r, K^iz^i \in \mathcal{U}^i\\
&\mathcal{X}^i_r \subset \mathcal{X}^i
\end{align}  
\indent Where $K^i$ is chosen as  the unconstrained LQ optimal feedback control law for the nominal system$(5)$, while $P^i$ is the solution of the Lyapunov equation
\begin{align}
(A^i+B^iK^i)^T P^i (A^i+B^iK^i) + Q^i + {K^i}^TR^iK^i = P^i
\end{align}
And then we give the terminal set a define $\mathcal{X}^i_\varepsilon \triangleq \{z^i \in R^{n_i}| \parallel z^i \parallel_{P^i} \le \varepsilon^i\}$ 
with $\sqrt{1-\frac{\lambda_{min}(Q^i)}{\lambda_{max}(P^i)}} r^i \le \varepsilon^i \le r^i$. 
\quad We'll use this condition to guarantee recursive feasibility and stability under a self-triggered mechanism.
\begin{lemma}\citep{DAI20191446}
	The nominal predictive state difference between $t_k$ and $t_{k+1} = t_k + M_k^i$ is norm-bounded by
	\begin{align}
	||z^i(t_{k+1}+l|t_{k+1})-z^i(t_{k+1}+l|t_k)||_\phi &= \parallel{A^i}^le^i(t_k+M_k^i|t_k)\parallel_\phi\notag \\
	&\le
	\begin{cases}
	\mathcal{B}(\centerdot)||A^i||^l(1-||A^i||^{M_k^i}), & \text{if} ||A^i||\neq 1\\
	\mathcal{B}(\centerdot)M_k, & \text{if} ||A^i||= 1
	\end{cases}
	\end{align}
	where $\phi\in\mathbb{S}_{++}^{n}$ and $\mathcal{B}(\centerdot):\mathbb{S}_{++}^{n}\to\mathbb{R}$ is defined as
	\begin{equation}
	\mathcal{B}(\centerdot)\triangleq
	\begin{cases}
	\sqrt{\lambda_{max}(\centerdot)}\frac{\bar{w^i}}{1-||A^i||},& \text{if} ||A^i||\neq 1\\
	\sqrt{\lambda_{max}(\centerdot)}\bar{w^i}, & \text{if} ||A^i||= 1
	\end{cases}
	\end{equation}
\end{lemma}
\begin{lemma}\citep{DAI20191446}
	For any sampling instant $t_k$ $k\in N$, let $u^*(t_k)$ be the solution of $\mathbb{P}(x^i(t_k))$. If the first $M_k^i$ step are applied to $(1)$ in a strictly open-loop fashion, then the difference between $J^*(x^i(t_k),u^i(t_k))$ and $J^*(x^i(t_{k+1}),u^i(t_{k+1}))$ is bounded by
	\begin{align}
	&J^*(x^i(t_{k+1}),u^i(t_{k+1})) - J^*(x^i(t_k),u^i(t_k))\le g^i(M^i_k,x^i(t_k),u^{i,*}(t_k),\bar{w^i}) \\
	\text{where} &\notag\\
	&g^i(M^i_k,x^i(t_k),u^*(t_k),\bar{w^i}) \triangleq g_0^i(M^i_k,x^i(t_k),u^*(t_k),\bar{w^i}) -\notag\\
	& \sum_{l=0}^{M_k^i-1}[||z^*(t_k+l|t_k)||_{Q^i}^2 + ||u^*(t_k+l|t_k)||_{R^i}^2].
	\end{align}
\end{lemma}
In discrete system ,we get a self-triggering machanism which is proposed as following 
\begin{subequations}
	\begin{align}
	M^i_k \triangleq \arg\min_{M_k\in N_{[1,N]}}g^i(M_k,x^i(t_k),u^*(t_k),\bar{w^i})\\
	\text{subjec to}\quad g^i(M_k,x^i(t_k),u^*(t_k),\bar{w^i}) < 0
	\end{align}
\end{subequations}
\indent If the optimizaition problem is infeasible, then we set $M^i_k = 1$.

\section{RDMPC scheme for the coupled constraints system}
Our control goal is to design a robust self-triggered DMPC algorithm, which can stablize the whole system at a fast convergence rate and satisfied with constraints (2) and (3), while simultaneously reducing the amount of compute burden.
\subsection{Constriants tightening }
Now considering the coupled constraints$(3)$, as we construct a nominal system$(5)$ which has a tube error $e$ with actual system$(1)$, we should tightening the constraints to satisfied with the feasibility. Specifically, the tightening constraints of $(3)$ are
\begin{subequations}
	\begin{align}
	\sum_{i=1}^{M}(\Psi^i_xz^i(t_k+l|t_k)+\Psi^i_uu^i(t_k+l|t_k))\leq(1- \epsilon(l))\textbf{1}_p,
	\quad &\forall l \in \mathbb{N}_{[0,N-1]}\\
	\sum_{i=1}^{M}(\Psi^i_Nz^i(t_k+N|t_k))\leq(1- \epsilon(N))\textbf{1}_p,\quad & \forall z^i(t_k+N|t_k) \in T^i_f
	\end{align}
\end{subequations}
where terminal coupled constriants coefficient matrix formulate as $\Psi^i_N = \Psi^i_x +\Psi^i_uK^i$, and  supposed there are $I$ subsystems $||A^i||\neq 1,$ $J$ subsystems $||A^j||= 1$ in the discrete system. We give a simlified denote that
\begin{align}
&\epsilon(l) =  \notag 
\sum_{i=1}^{I}\parallel\Psi_x^i\parallel\bar{w^i}\frac{1-\parallel A^i \parallel^l}{1-\parallel A^i \parallel} + \sum_{j=1}^{J}\parallel\Psi_x^j\parallel l\bar{w^j}
\end{align}
\begin{align}
&\epsilon(N) =  \notag 
\sum_{i=1}^{I}\parallel\Psi_N^i\parallel\bar{w^i}\frac{1-\parallel A^i \parallel^N}{1-\parallel A^i \parallel} + \sum_{j=1}^{J}\parallel\Psi_N^j\parallel N\bar{w^j}\notag
\end{align}
\indent Obviously, the tolerance $0 < \epsilon(1) < \dots < \epsilon(N) <1$ should be satisfied to ensure Gradual decrease in terminal constraints and $0 \in T^i_f$.
Correspondingly, the tightened RDMPC formulation is:
\begin{align}
\mathbb{P}_{\epsilon}(x^i(t_k)) \triangleq \min_{\textbf{u}^i}\{\sum_{i=1}^{M}J^i(x^i(t_k),\textbf{u}^i(t_k)):(6b-6f) \text{and} (15a)\}
\end{align}
$\textbf{u}^i := \{u^i(0),u^i(1),\dots,u^i(N-1)\}$,
the choice of $T^i_f$ is chosen to be$	T^i_f \triangleq \mathcal{X}^i_\varepsilon$, which subject to
\begin{align}
\sqrt{1-\frac{\lambda_{min}(Q^i)}{\lambda_{max}(P^i)}} r^i \le \varepsilon^i \le r^i.
\end{align}
\indent We can represented the tightend RDMPC formulation as
\begin{subequations}
	\begin{align}
	\mathbb{P}_\epsilon(x(t_k)):  &\min_{u^i(k)}\sum_{i=1}^{M}J^i(x^i(t_k),\textbf{u}^i(t_k))\\
	s.t. 
	& (6b)-(6f), i \in [1,M] \notag\\
	& \sum_{i=1}^{M}f^i(x^i(t_k),\textbf{u}^i(t_k)) \le b(\epsilon)
	\end{align}
\end{subequations}
where $f^i(x^i(t_k),\textbf{u}^i(t_k))$ is an appropriate function form rewriting from $(15a)$, $b(\epsilon) =  [\textbf{1}_p,(1- \epsilon(1))\textbf{1}_p,\dots,(1- \epsilon(N-1))\textbf{1}_p]^T$.\\
\indent The connection of nominal system and real system is shown in the following lemma 3.1, while its proof locate at appedix A.
\begin{lemma}
	The coupling constraints (3) in real system can be satisfied when (15) is satisfied in nominal system if the upper bound of disturbance satisfies with
	\begin{align}
	\begin{cases}
	\bar{w^i}\le (\frac{1}{M\parallel\Psi_N^i\parallel}-\frac{r^i}{\sqrt{\lambda_{max}(P^i)}})\frac{1-\parallel A^i\parallel}{1-\parallel A^i\parallel^N}, & \text{if} \parallel A^i\parallel\neq 1\notag\\
	\bar{w^j}\le (\frac{1}{M\parallel\Psi_N^j\parallel}-\frac{r^j}{\sqrt{\lambda_{max}(P^j)}})\frac{1}{N}  &\text{if} \parallel A^j\parallel= 1
	\end{cases}
	\end{align}
\end{lemma}
\indent In the end of this subsection, we give some general assumptions to ensure initial feasibility of OCP and global constraints from predictive system to real system.\\
$(A1)$: $(A^i,B^i)$ is reachable and $x^i(t_k)$ is measurable for all $i \in \mathbb{Z}_M$.\\
$(A2)$: $\mathcal{W}^i$,  $\mathcal{X}^i$ and $\mathcal{U}^i$ are polytope containing the origins in its interior for all $i \in \mathbb{Z}_M$.\\
$(A3)$: The undirected graph $G = (\mathcal{V},\mathcal{E})$ is full connected.\\
\subsection{Distributed ADMM form of OCP }
Let $\lambda \in\mathbb{R}^{p\times N}$ be the dual variable associated with constraints $(17b)$, the lagrange function of OCP can formulate as:
\begin{align}
\mathcal{L}({\textbf{u}^i},\lambda) =: \sum_{i=1}^{M}J^i(x^i,\textbf{u}^i) + \lambda^T(\sum_{i=1}^{M}f^i(x^i,{\textbf{u}^i})-b(\epsilon))
\end{align}
The dual problem is
\begin{align}
\max_{\lambda\ge 0} \min_{\textbf{u}^i \in \mathcal{U}} \mathcal{L}({\textbf{u}^i},\lambda) := \min_{\lambda\ge 0} \max_{\textbf{u}^i \in \mathcal{U}} -\mathcal{L}({\textbf{u}^i},\lambda) =\min_{\lambda\ge 0}\sum_{i=1}^{M}g^i(\lambda)
\end{align}
where
\begin{align}
g^i(\lambda) := \max_{\textbf{u}^i \in \mathcal{U}}-J^i(x^i,\textbf{u}^i)-\lambda(f^i(x^i,{\textbf{u}^i})-\frac{1}{M}b(\epsilon))
\end{align}
The dual problem is not distributed as $\lambda$ is a common variable in $g^i(\lambda)$. According to assumption $(A3)$, the dual problem $(19)$ can rewritten as a consensus problem in the following
\begin{align}
\min_{\lambda^{i}\ge 0}\sum_{i=1}^{M}g^i(\lambda^i) \quad s.t. \lambda^i =\lambda^j, (i,j)\in \mathcal{E}
\end{align}
where $\lambda^i$ is the local copy of $\lambda$ in the $i^{th}$ system , and the conditions $\lambda^i =\lambda^j$ ensure the consensus of whole system dual variable. It can be further rewritten by using a new set of reference variable in the form of 
\begin{align}
\max_{\lambda^i\ge 0}\sum_{i=1}^{M}g^i(\lambda^i) \quad s.t. \sum_{i=1}^{M}E^i\lambda^i = c.
\end{align}
The augmented lagrangian methods form of $(23)$ is
\begin{align}
\mathcal{L}_{\rho}(\lambda^i,\omega) 
=&\sum_{i=1}^{M}g^i(\lambda^i) + \omega^T(\sum_{i=1}^{M}E^i\lambda^i - c) + \frac{\rho}{2}\parallel \sum_{i=1}^{M}E^i\lambda^i - c \parallel^2_2.
\end{align}
\indent To apply the ALM form to solve problem$(22)$, the distributed ADMM is giving in the following.
\begin{align}
&(\lambda^1_{k+1},\dots,\lambda^M_{k+1})=\arg\min_{\lambda^i\ge 0}\mathcal{L}_{\rho}(\lambda^1,\dots,\lambda^M,\omega_k),\\
&\textbf{u}^i_{k+1}=\arg\min_{{\textbf{u}^i \in \mathcal{U}^i}}\mathcal{L}({\textbf{u}^i},\lambda^i_{k+1})\\
&\omega_{k+1}=\omega_k-\rho(\sum_{i=1}^{M}E^i\lambda^i_{k+1} - c)
\end{align}
\indent For every subproblem with regard to $\lambda^i$, we can introduce neighbor term $\frac{1}{2}\parallel \lambda^i - \lambda^i_k \parallel_s^2$ and relaxation factor $\gamma$ to reduce iteration burden. Then the updating step of $\lambda^i_{k+1},\omega_{k+1}$ comes as
\begin{align}
&\lambda^i_{k+1}=\arg\min_{\lambda^i\ge 0}\mathcal{L}_{\rho}(\lambda^1_k,\dots,\lambda^i,\dots,\lambda^M_k,\omega_k)+\frac{1}{2}\parallel \lambda^i - \lambda^i_k \parallel_s^2,\\
&\omega_{k+1}=\omega_k-\rho\gamma(\sum_{i=1}^{M}E^i\lambda^i_{k+1} - c)
\end{align}
\indent We noted that $\textbf{u}^i_{t_k} = \{u^i_{k}(0),u^i_{k}(1),\dots,u^i_{k}(N-1)\}$ is the optimal solution of $\mathbb{P}_\epsilon(x(t_k))$,and define the applied optimal control sequence as $\textbf{u}^{i,*}(t_k) \triangleq  \{ u^i_{k}(0),u^i_{k}(1),\dots,u^i_{k}(M_k) \} $ when we get $M_k$ from algorithm 2. Then we give the closed-loop sysytem under the Robust self-triggered DMPC before next sampling time in the following:
\begin{align}
x^i(t_k+l+1) = A^ix^i(t_k+l) + B^iu^{i,*}(t_k+l) + w^i(t_k + l)
\end{align}
\subsection{Self-triggered DMPC algorithm}
The overall procedure of the distributed ADMM algorithm at time $t_k$ is summarized in the following algorithm 1.
\begin{algorithm}
	\caption{Consensus ADMM algorithm} 
	\hspace*{0.02in} {\bf Input:} 
	Input measured system state  $x^i$,$i\in\mathbb{Z}_M$\\
	\hspace*{0.02in} {\bf Output:} 
	$\textbf{u}^{i,*}_{k},i\in\mathbb{Z}_M$\\
	Initiallization: choose $\rho > 0$,set $k=0,\lambda_0^i=0,\omega_0 = 0$ for all $i\in\mathbb{Z}_M,j\in N_i$.\\
	\textbf{repeat}\\
	\textbf{for}all $i\in\mathbb{Z}_M$(parallel)\textbf{Do}\\
	obtain $\lambda^i_{k+1}$ from (28);\\
	obtain $\textbf{u}^i_{k+1}$ from (26);\\
	obtain $\omega_{k+1}$ from (29);\\
	\textbf{end for} $k \leftarrow k+1$
	\textbf{until} stop criterion is satisfied
\end{algorithm}\\
\indent From \citep{DAI20191446} we know that self-triggered machanism can reduce the burden on compuation largely, we use it in the overall Robust self-triggerd DMPC scheme at time $t_k$. The framework is shown in the algorithm 2.
\begin{algorithm}
	\caption{Robust self-triggered DMPC }
	1:Every subsystem $i$ measure its own state $x^i(t_k)$;\\
	\textbf{IF} $x^i(t_k) \in \mathcal{X}^i_{\varepsilon^i}$, Subsystem $i$ obtains local feedback control law $u^i(t_k) = K^ix^i(t_k)$ and applied itself until maximum iteration step.\\
	\textbf{IF ELSE}\\
	2:Other subsystem $i$ calls algorithm 1 with $x^i(t_k)$ and get $\textbf{u}^{i,*}$;\\
	3:Subsystem $i$ calculate optimization problem (14a) to get $M^i_k $;\\
	4:Choose $M_k = \min\{M_k^i\}$ as next sampling step;\\
	5:Subsystem $i$ obtains $\textbf{u}^{i,*}(x(t_k))$ and applied it to Every subsystem $i$ ;\\
	6:Let $t_k = t_k+M_k$ and go back to step 1.
\end{algorithm}
\section{Main properties of the RDMPC scheme}
This section presents the main properties of RDMPC schme in algorithm 2, that is recursive feasibility, satisfaction of  both the local state and input constraints and the global coupled constraints, and closed-loop stability.
The following lemma ensure that if OCP $(18)$ is feasible at some sampling time, then its feasiblity remains at the next sampling time.
\begin{lemma}[Recursive feasibility of nominal predictive system]
	Assume that $\bm{u}(t_k)=[\bm{u}^1(t_k), \bm{u}^2(t_k), \dots, \bm{u}^M(t_k)]$ where $\bm{u}^i(t_k) = [u^i(t_k|t_k), u^i(t_k+1|t_k), \dots, u^i(t_k+N-1|t_k)]^T$ is a feasible solution at time $t_k$ for OCP $\mathbb{P}_\epsilon(x(t_k))$, 
	Then the OCP $\mathbb{P}_\epsilon(x(t_{k+1}))$ is feasible at time $t_{k+1}$.\\
	If the bound on disturbance satisfies with \\
	(1)the local condition 
	\begin{align}
	\begin{cases}
	\bar{w^i}\le \frac{(r^i - \varepsilon^i)(1-||A^i||)}{\sqrt{\lambda_{max}(P^i)}(1-||A^i||)^N}, & \text{if} ||A^i|| \neq 1\notag\\
	\bar{w^j}\le \frac{(r^j - \varepsilon^j)}{\sqrt{\lambda_{max}(P^j)}N},  &\text{if} ||A^j||= 1
	\end{cases}
	\end{align}
	(2)the global condition
	\begin{align}
	\begin{cases}
	\bar{w^i}\le (\frac{1}{M\parallel\Psi_N^i\parallel}-\frac{r^i}{\sqrt{\lambda_{max}(P^i)}})\frac{1-\parallel A^i\parallel}{1-\parallel A^i\parallel^N}, & \text{if} \parallel A^i\parallel\neq 1\notag\\
	\bar{w^j}\le (\frac{1}{M\parallel\Psi_N^j\parallel}-\frac{r^j}{\sqrt{\lambda_{max}(P^j)}})\frac{1}{N}  &\text{if} \parallel A^j\parallel= 1
	\end{cases}
	\end{align}\\
\end{lemma} 
The proof of lemma 4.1 locate at appedix B. Through the recursive feasibility lemma of nominal system, we can conclude the theorem of recursive feasibility and constraints satisfaction in real system.
\begin{theorem}[Recursive feasibility and constraints satisfaction]
	If the $\mathbb{P}_\epsilon(x(t_0))$ is feasible and the conditions in lemma 4.1 is satisfied, then for the distributed system $(1)$ under algorithm 2, it holds that \\
	(1)$\mathbb{P}_\epsilon(x(t_k))$  is feasible for all $t_k$ and $k \in \mathbb{N}$;\\
	(2)$x^i(t)\in \mathcal{X}^i, u^i(t)\in \mathcal{U}^i$,  $\sum_{i=1}^{M}(\Psi^i_xx^i(t)+\Psi^i_uu^i(t))\leq\textbf{1}_p$ for every realization $w^i(t)\in \mathcal{W}^i.$
\end{theorem}
\begin{proof}
	By lemma 4.1, the assumption that feasible of initial optimization problem $\mathbb{P}_\epsilon(x(t_0))$  implies all $\mathbb{P}_\epsilon(x(t_k))$  is feasible for all $t_k$ and $k \in \mathbb{N}$, which proves (1). \\
	\text{\quad} To show the satisfaction of local constriants, the time horizon is divided into two parts, $k\in\mathbb{N}_{[0, t^i]}$ and $k\in\mathbb{N}_{[ t^i, \infty]}$. Where $t^i$ denotes the time when the state of system $i$ first enter $\mathcal{X}^i_{\varepsilon}$. for every sampling time $t_k$ before the time $t^i$, from (30) and (15), for $ l \in \mathbb{N}_{\le M_k}$, we have the close-loop system formulat as 
	\begin{align}
	x^i(t_k+l+1) &= A^ix^i(t_k+l) + B^iu^{i,*}(t_k+l) + w^i(t_k + l)\notag\\
	&= A^i(z^{i,*}(t_k+l|t_k) + e^i(t_k+l|t_k)) + B^iu^{i,*}(t_k+l|t_k)+ w^i(t_k + l)\notag\\
	&= z^{i,*}(t_k+l+1|t_k) + e^i(t_k+l+1|t_k)\notag\\
	&\in \mathcal{Z}^i_{l+1} \oplus \mathcal{R}^i_{l+1} \subseteq \mathcal{X}\notag
	\end{align}
	\indent Further, according to constraints $(6f)$, $u^i(t_k+l+1) = u^{i,\star}(t_k+l+1) \in \mathcal{U}^i$.\\
	\indent Now considering the sampling time after $t^i$, the state of system $i$ have already enter $\mathcal{X}^i_{\varepsilon}$. The local control $u^i(t^i+l) = K^ix^i(t^i+l) \in \mathcal{U}^i, l \in \mathbb{N}$ is applied to the real system as dual control strategy.
	So the close-loop system $i$ after $t^i$ is
	\begin{align}
	x^i(t^i+l+1) &= A^ix^i(t^i+l) + B^iK^ix^i(t^i+l) + w^i(t^i+l)\notag\\
	&= (A^i + B^iK^i)x^i(t^i+l) + w^i(t^i+l)\notag
	\end{align}
	from the definition of RCIS $\mathcal{X}^i_{\varepsilon}$, $x^i(t^i) \in \mathcal{X}^i_{\varepsilon}$ implies $(A^i + B^iK^i)x^i(t^i) \in \mathcal{X}^i_{\varepsilon}, l \in \mathbb{N}$, according to the relation of disturbance and RCIS, $  \lambda_{max}(P^i)\bar{w^i} \le r^i - \varepsilon^i$, it holds from (8) that $	x^i(t^i+1) \in \mathcal{X}^i_{r} \in \mathcal{X}$. Consider $x^i(t^i+l)\in \mathcal{X}^i_{r}$ and (B3) we have 
	\begin{align}
	&\parallel(A^i + B^iK^i)x^i(t^i+l)\parallel_{P^i} \le \varepsilon^i\notag\\
	&\parallel x^i(t^i+1+1)\parallel_{P^i} \le \parallel(A^i + B^iK^i)x^i(t^i+l)\parallel_{P^i} +\parallel w^i(t^i)\parallel_{P^i}\notag\\
	&\le\varepsilon^i + r^i - \varepsilon^i=r^i\notag 
	\end{align}
	\indent Then $x^i(t^i+l)\in\mathcal{X}^i_{r}\in\mathcal{X}^i, u^i(t^i+l) = K^ix^i(t^i+l) \in \mathcal{U}^i, l \in \mathbb{N} $. \\
	\indent We denote $t_{out}=\min\{t^i\}$, $t_{in}=\max\{t^i\}$, for every sampling time $t_k$ before the time $t_{out}$, according to constraints $(15a)$ and lemma 3.1, we have 
	\begin{align}
	\sum_{i=1}^{M}(\Psi^i_xx^i(t_k)+\Psi^i_uu^i(t_k))\leq\textbf{1}_p\notag
	\end{align}	
	\indent For every sampling time $t_k+l$ between  $t_{out}$ and $t_{in}$, assume L subsystem have already get in terminal set. As we take $u^i(t_k+l)=u^i(t_k+l|t_k)$ and $u^i(t_k+l|t_k)=K^ix^i(t_k+l|t_k)$ when $x^i(t_k+l|t_k)\in\mathcal{X}^i_{r}$, according to constraints $(15a)$ and lemma 3.1, then we have 
	\begin{align}
	&\sum_{i=1}^{M-L}(\Psi^i_xx^i(t_k+l)+\Psi^i_uu^i(t_k+l))+ \sum_{i=1}^{L}(\Psi^i_xx^i(t_k+l)+\Psi^i_uK^ix^i(t_k+l))\notag\\
	&=\sum_{i=1}^{M-L}(\Psi^i_xx^i(t_k+l|t_k)+\Psi^i_uu^i(t_k+l|t_k))+ \sum_{i=1}^{L}(\Psi^i_xx^i(t_k+l|t_k)+\Psi^i_uK^ix^i(t_k+l|t_k))\notag\\
	&=\sum_{i=1}^{M-L}(\Psi^i_xx^i(t_k+l|t_k)+\Psi^i_uu^i(t_k+l|t_k))+ \sum_{i=1}^{L}(\Psi^i_xx^i(t_k+l|t_k)+\Psi^i_uu^i(t_k+l|t_k))\notag\\
	&=\sum_{i=1}^{M}(\Psi^i_xx^i(t_k+l|t_k)+\Psi^i_uu^i(t_k+l|t_k))\le\textbf{1}_p\notag
	\end{align}
	\indent For every sampling time $t_k$ after the time $t_{in}$, according to last conclusion in lemma 3.1, we have 
	\begin{align}
	\sum_{i=1}^{M}\Psi^i_Nx^i(t_{k}) \le (1- \epsilon(N)) \textbf{1}_p  \le\textbf{1}_p \notag
	\end{align}
	Hence, we have already  proved (2) for different situation of all sampling time. 
\end{proof}
In the following part, we concentrate our attentions to the stability of whole distributed system. Fisrt we give a normal definition of stability in discrete-time system.
\begin{definition}
\citep{Sontag199520}	System $x(k+1)=f(x(k),u(k))$ is (globally) input-to-state stable(ISS) if there exist a $\mathcal{KL}$-function $\beta : \mathbb{R}_{\ge 0} \times \mathbb{R}_{\ge 0} \to \mathbb{R}_{\ge 0}$ and a $\mathcal{K}$-function $\gamma$ such that, for each input $u \in l_{\infty}^m $ and each $\xi \in \mathbb{R}^n $, it holds that
\begin{align}
x(k,\xi,u) \le \beta(k,|\xi|) + \gamma(u)
\end{align}
for each $ t \in \mathbb{Z}^n$.
\end{definition}
When the dynamics comes to a closed-loop perturbed system, we introduce a common lemma to simplify the prove of ISS.
\begin{lemma}
\citep{JIANG2001857} A system of the form $x(k+1)=f(x(k),w(k))$ is input-to-state stability(ISS) if and only if there exists a continuous ISS-lyapunov fuction $V(x(k))$ such that for $\mathcal{K}_\infty$ functions $\rho_1(\centerdot),\rho_2(\centerdot)$, and $\rho_3(\centerdot)$, and a $\mathcal{K}$ function $\alpha(\centerdot)$, it satisfies with 
\begin{align}
&\rho_1(\parallel x(k)\parallel) \le V(x(k)) \le \rho_2(\parallel x(k)\parallel) \notag\\
&V(x(k+1)) - V(x(k)) \le \alpha(\bar{w}) - \rho_3(\parallel x(k)\parallel) \notag
\end{align}  
\end{lemma}
Then we give the input-to-state stability theorem of RDMPC scheme as follows.
\begin{theorem}
Given fesibility of $\mathbb{P}_\epsilon(x(t_0))$ and satisfaction of lemma 4, the closed-loop system in (30) is ISS.
\end{theorem}
\begin{proof}
The ISS-Lyapunov function can be chosen as $V(\textbf{x}(t_k)) \triangleq J^*(\textbf{x}(t_k),\textbf{u}^*(t_k))$, apparently $J^*(\textbf{x}(t_k),\textbf{u}^*(t_k)) = \sum_{i=1}^{M}J^*(x^i(t_{k}),u^i(t_k))$ is a $\mathcal{K}_\infty$ functions where $\textbf{x}(t_k)) = [x^1(t_{k}),x^2(t_{k}),\dots,x^M(t_{k})]$, $\textbf{u}^*(t_k)) = [u^{1,*}(t_{k}),u^{2,*}(t_{k}),\dots,u^{M,*}(t_{k})]$ is boundness in local constriants (2), so the first ISS condition is trivially proved. Then from lemma 2,the difference between $V(\textbf{x}(t_k))$ and $V(\textbf{x}(t_k+1))$ is bounded by
\begin{align}
V(\textbf{x}(t_k+1)) - V(\textbf{x}(t_k))&= \sum_{i=1}^{M}\{J^*(x^i(t_{k+1}),u^i(t_{k+1})) - J^*(x^i(t_k),u^i(t_k))\}\notag\\
&\le \sum_{i=1}^{M}\{g^i(M^i_k,x^i(t_k),u^{i,*}(t_k),\bar{w^i})\}\notag\\
&\le \sum_{i=1}^{M}\{g_0^i(M^i_k,x^i(t_k),u^*(t_k),\bar{w^i}) - [||x^i(t_k)||_{Q^i}^2 + ||u^{i,*}(t_k)||_{R^i}^2]\}\notag \\
&\le \sum_{i=1}^{M}\{\beta(\bar{w^i}) - ||x^i(t_k)||_{Q^i}^2\}\notag\\
&\le \alpha(\bar{\textbf{w}^i}) - \sum_{i=1}^{M}||x^i(t_k)||_{Q^i}^2\notag
\end{align}
where the bound of  distrubance in real system defined as $\bar{\textbf{w}^i} = [\bar{w^1},\bar{w^2},\dots,\bar{w^M}]$, while $\alpha(\bar{\textbf{w}^i}) \triangleq \sum_{i=1}^{M}\beta(\bar{w^i})$ is a $\mathcal{K}$ function defined as the maximum of $\sum_{i=1}^{M}g^i(M^i_k,x^i(t_k),u^{i,*}(t_k),\bar{w^i})$ and $\sum_{i=1}^{M}||x^i(t_k)||_{Q^i}^2$ is absolutely $\mathcal{K}_\infty$ functions on $\textbf{x}(t_k)$.
\end{proof}
\section{Numerical simulation}
The numerical example is a 4-agent system with same system matrix $A^{1,2,3,4} = [1.1, 0.12;0.35 ,0.0075]$, $B^{1,2,3,4} = [1.5;0.5]$, the local constraints in all subsystems are $\mathcal{X}^i = \{x^i||x^{i,1}|\le20 , |x^{i,1}| \le 5 \}$ and $\mathcal{U}^i = \{u^i||u^i|\le 2\}$, while global constraints is $\parallel \sum_{i=1}^4\Psi^i_xx^i+0.01u^1 + 0.02u^2 + 0.03u^3 + 0.04u^4\parallel \le 10\textbf{1}_p$. Where $\Psi^i_x=[0.08,0.02],p=1$. The set of bounded disturbance are satisfied with condition of lemma 3, we set that as $\mathcal{W}^i = [-0.3,0.3] \times [-0.3,0.3]$. Set weight matrix $Q^i = [1,0;0,1],R^i = 0.1, i =1,2,3,4$, we obtain from LQR feedback control law and Riccati equation that $P^{1,2,3,4} = [1.0516, 0.0057;0.0057, 1.0015], K^{1,2,3,4} = [-0.7033,-0.0710]$, and initiall state $x^1(0) = [-19;-4], x^2(0) = [-18;-3], x^3(0) = [-10;4], x^4(0) = [-18;3]$. The prediction horizon is chosen to be $N =5$ and simulation length is $T_{run}=30$ steps.
\begin{figure}[htbp]
\centering{\includegraphics[width=3.5in,height=2.5in]{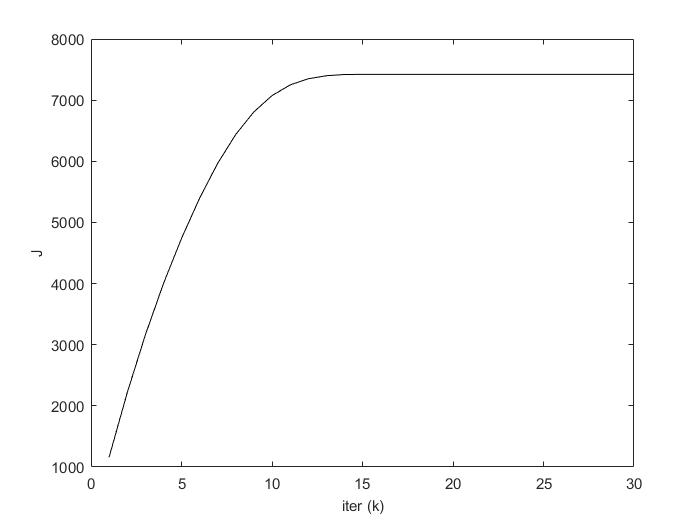}}
\caption{Cost function of overall system under algorithm 2}
\label{fig}
\end{figure}
\begin{figure}[htbp]
\centering{\includegraphics[width=3.5in,height=2.5in]{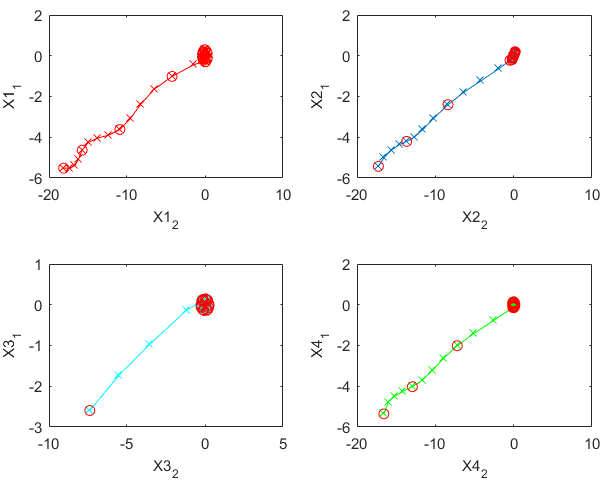}}
\caption{System trajectories under algorithm 2}
\label{fig}
\end{figure}
\begin{figure}[htbp]
\centering{\includegraphics[width=3.5in,height=2.5in]{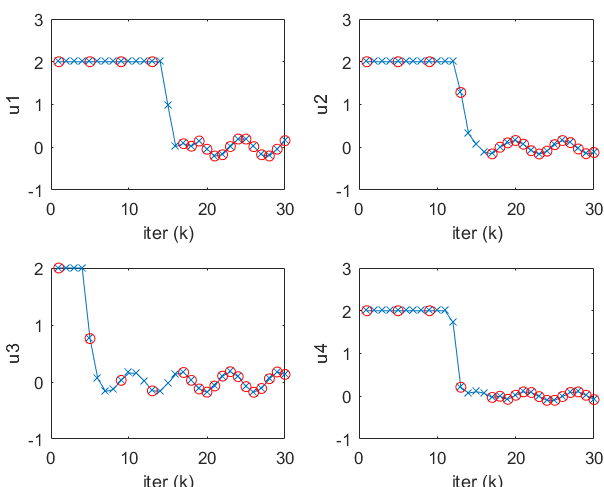}}
\caption{Control inputs trajectories under algorithm 2}
\label{fig}
\end{figure}
\begin{figure}[htbp]
\centering{\includegraphics[width=3.5in,height=2.5in]{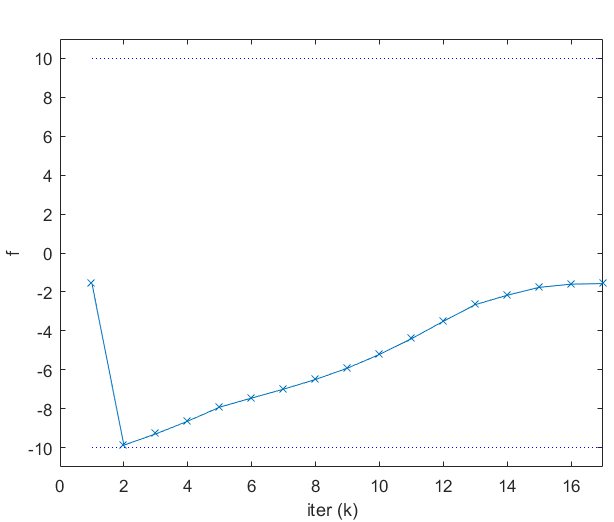}}
\caption{Global constraints of overall system}
\label{fig}
\end{figure}
\indent Fig.1 indicates the summarize of the cost function in all subsystem, from that we can know the objective function in $(6a)$ is convergence to a stable value within algorithm tolerance.\\ 
\indent Fig.2 shows the sate trajectories of all subsystem under algorithm 2 with different sequence of disturbances, the system state convergence to the neighborhood of the origin, the state trajectories reflect the excellent performance of algorithm we proposed from the stability side.\\ 
\indent Fig.3 dispalys the Control inputs trajectories under algorithm 2, the sampling instants is indicated by scarlet circles while each iteration steps denotes by blue crossmark. Iteration time of algorithm without and with self-triggering machanism is $T_{run}^{a1}=1.8190s$ and $T_{run}^{a2}=1.1097s$, the difference between computation time shows off the latter algorithm can lossen computation burden apparently. \\
\indent The evolution of global constraints function $f=\sum_{i=1}^{M}f^i$ has shown in Fig.4, in which the blue dashed line represents the maximum value of the global constraint. Together with Fig.2 and Fig.3, they illustrated the satisfaction of the local constraints on state and control input for every subsystem, as well as the global constraints in overall system.

\section{conclusions}
A self-triggered DMPC scheme has been proposed for disturbed linear systems with local and global constraints. The proposed scheme uses tube method to capture the bounded disturbance in subsystem, while it also gives a form of global constraints tightening based on upper bound of disturbance to apply parallel distributed ADMM methods. For closed-loop system under algorithm 2, the properties about the recursive feasibility of OCP and ISS stability of overall system are given, some numerical examples has been shown to demonstrated the good performance of proposed scheme.

\section*{Acknowledgement(s)}
An unnumbered section, e.g.\ \verb"\section*{Acknowledgements}", may be used for thanks, etc.\ if required and included \emph{in the non-anonymous version} before any Notes or References.

\section{References}
\bibliographystyle{apacite}
\bibliography{RDMPC1}
\section{Appendices}

\appendix
\section{the proof of lemma 3.1}
\begin{proof}
From the proof of lemma 1\citep{DAI20191446}, we know the error $e$ between nominal and real system satisfies with:
\begin{align}
&\begin{cases}
\parallel e^i(t_k+l|t_k)\parallel \le \frac{1-||A^i||^l}{1-||A^i||}\bar{w^i},||A^i|| \neq 1\notag\\
\parallel e^j(t_k+l|t_k)\parallel \le l\bar{w^j},||A^j|| = 1 \notag\\
\end{cases}
\end{align}
Then we can get
\begin{align}
&\parallel\sum_{i=1}^{M}\Psi^i_xe^i(t_k+l|t_k)\parallel \le \sum_{i=1}^{M}\parallel\Psi^i_x\parallel\parallel e^i(t_k+l|t_k)\parallel\notag\\
&\le\sum_{i=1}^{I}\parallel\Psi_x^i\parallel\bar{w^i}\frac{1-\parallel A^i \parallel^{l}}{1-\parallel A^i \parallel} + \sum_{j=1}^{J}\parallel\Psi_x^j\parallel l\bar{w^j}
=\epsilon(l)\notag\\
&\sum_{i=1}^{M}\Psi^i_xe^i(t_k+l|t_k)\le\epsilon(l)\textbf{1}_p
\end{align}
According to the condition about upper bound of disturbance, we have
\begin{align}
&\begin{cases}
\bar{w^i}\le (\frac{1}{M\parallel\Psi_N^i\parallel}-\frac{r^i}{\sqrt{\lambda_{max}(P^i)}})\frac{1-\parallel A^i\parallel}{1-\parallel A^i\parallel^N}, & \text{if} \parallel A^i\parallel\neq 1\notag\\
\bar{w^i}\le (\frac{1}{M\parallel\Psi_N^j\parallel}-\frac{r^j}{\sqrt{\lambda_{max}(P^j)}})\frac{1}{N}  &\text{if} \parallel A^j\parallel= 1
\end{cases}\notag\\
&\begin{cases}
\parallel\Psi_N^i\parallel(\bar{w^i}\frac{1-\parallel A^i\parallel^N}{1-\parallel A^i\parallel}+\frac{r^i}{\sqrt{\lambda_{max}(P^i)}})\le \frac{1}{M}, & \text{if} \parallel A^i\parallel\neq 1\notag\\
\parallel\Psi_N^j\parallel(N\bar{w^j}+\frac{r^j}{\sqrt{\lambda_{max}(P^j)}})\le \frac{1}{M}, & \text{if} \parallel A^j\parallel= 1\notag\\
\end{cases}\notag\\
&\sum_{i=1}^{I}(\parallel\Psi_N^i\parallel\bar{w^i}\frac{1-\parallel A^i\parallel^N}{1-\parallel A^i\parallel}+\parallel\Psi_N^i\parallel\frac{r^i}{\sqrt{\lambda_{max}(P^i)}})+\notag\\
&\sum_{j=1}^{J}(\parallel\Psi_N^j\parallel N\bar{w^j}+\parallel\Psi_N^j\parallel\frac{r^j}{\sqrt{\lambda_{max}(P^j)}}) \le 1\notag\\
&\sum_{i=1}^{M}\parallel\Psi_N^i\parallel\frac{r^i}{\sqrt{\lambda_{max}(P^i)}}+(\sum_{i=1}^{I}\parallel\Psi_N^i\parallel\bar{w^i}\frac{1-\parallel A^i\parallel^N}{1-\parallel A^i\parallel}+\sum_{j=1}^{J}\parallel\Psi_N^j\parallel N\bar{w^j}) \le 1\notag\\
&\sum_{i=1}^{M}\parallel\Psi_N^i\parallel\frac{r^i}{\sqrt{\lambda_{max}(P^i)}}+\epsilon(N)
\le 1\\
&\parallel\sum_{i=1}^{M}\Psi^i_Ne^i(t_k+N|t_k)\parallel \le \sum_{i=1}^{M}\parallel\Psi^i_N\parallel \parallel e^i(t_k+N|t_k)\parallel = \epsilon(N) \notag\\
&\sum_{i=1}^{M}\Psi^i_Ne^i(t_k+N|t_k)\le\epsilon(N)\textbf{1}_p
\end{align}
When $l=0$, apparently we have 
\begin{align}
\sum_{i=1}^{M}(\Psi^i_xx^i(t_k|t_k)+\Psi^i_uu^i(t_k|t_k))=\sum_{i=1}^{M}(\Psi^i_xz^i(t_k|t_k)+\Psi^i_uu^i(t_k|t_k))\le \textbf{1}_p\notag
\end{align}
For $l \in Z_{[1,N-1]}$,
\begin{align}
&\sum_{i=1}^{M}(\Psi^i_xx^i(t_k+l|t_k)+\Psi^i_uu^i(t_k+l|t_k)) \notag\\
&= \sum_{i=1}^{M}(\Psi^i_x(z^i(t_k+l|t_k)+e^i(t_k+l|t_k))+\Psi^i_uu^i(t_k+l|t_k))\notag\\ 
&\leq (1- \epsilon(l))\textbf{1}_p + \sum_{i=1}^{M}\Psi^i_xe^i(t_k+l|t_k) \le \textbf{1}_p\notag
\end{align}
For $l = N $, from (A2) and (A3), we have
\begin{align}
&\parallel\sum_{i=1}^{M}\Psi^i_Nz^i(t_k+N|t_k)\parallel
\le\sum_{i=1}^{M}\parallel\Psi^i_N\parallel\parallel z^i(t_k+N|t_k\parallel\notag\\
&\le \sum_{i=1}^{M}\parallel\Psi^i_N\parallel\frac{\varepsilon^i}{\lambda_{max}(P^i)}\le \sum_{i=1}^{M}\parallel\Psi^i_N\parallel\frac{r^i}{\lambda_{max}(P^i)}\le  1- \epsilon(N)\notag\\
&\sum_{i=1}^{M}\Psi^i_Nx^i(t_k+N|t_k)=\sum_{i=1}^{M}\Psi^i_N(z^i(t_k+N|t_k)+e^i(t_k+N|t_k))\notag\\ 
&\le (1- \epsilon(N))\textbf{1}_p + \sum_{i=1}^{M}\Psi^i_Ne^i(t_k+N|t_k) \le \textbf{1}_p\notag
\end{align}
\end{proof}
\section{the proof of lemma 4.1}
\begin{proof}
Consider a candidate solution $\hat{\bm{u}}(t_{k+1})$ at time $	t_{k+1} $ given by
\begin{align}
\begin{cases}
\hat{u}^i(t_{k+1} + l|t_{k+1}) = u^i(t_{k+1} + l |t_k), &l\in\mathbb{N}_{[0, N-M^i_k-1]}\\
\hat{u}^i(t_{k+1} + l|t_{k+1}) = K^i\hat{z}^i(t_{k+1} + l|t_{k+1}), &l\in\mathbb{N}_{[N-M^i_k, N-1]}\notag
\end{cases}
\end{align}
where $\hat{z}^i(t_{k+1})$ is the prediction state of nominal system with $\hat{\bm{u}}(t_{k+1})$ and $x(t_{k+1})$.\\
Next we will prove $\hat{\bm{u}}(t_{k+1})$ is a feasible solution to $\mathbb{P}_\epsilon(x(t_{k+1}))$, i.e. it satisfies with (6b-6f) and (15).
Constraints (6b) and (6d) are obviouly satisfied in $\mathbb{P}_\epsilon(x(t_{k+1}))$.\\
For $l\in\mathbb{N}_{[0,N-M^i_k-1]}$, by the definition of $\hat{\bm{u}}(t_{k+1})$, (6f) can trivally obtained. From lemma 1 we get, 
\begin{align}
&\hat{z}^i(t_{k+1}+l|t_{k+1}) = z^i(t_{k+1}+l|t_k) + {A^i}^le^i(t_k+M_k^i|t_k) \notag\\
&\hat{z}^i(t_{k+1}+l|t_{k+1}) = z^i(t_k+M_k^i+l|t_k) + {A^i}^le^i(t_k+M_k^i|t_k) \notag
\end{align}
where we know from $(4)$ and $(5)$ that $z^i(t_k+M_k^i+l|t_k) \in \mathcal{Z}_{M_k^i+l}^i$, $e^i(t_k+M_k^i|t_k) \in \mathcal{R}_{M_k^i}^i$. It follows that
\begin{align}
\hat{z}^i(t_{k+1}+l|t_{k+1}) &\in \mathcal{Z}_{M_k^i+l}^i \oplus {A^i}^l\mathcal{R}_{M_k^i}^i \notag\\
&= \mathcal{X}^i \ominus \mathcal{R}_{M_k^i+l}^i \oplus {A^i}^l\mathcal{R}_{M_k^i}^i \notag\\
&= \mathcal{X}^i \ominus \mathcal{R}_l^i \triangleq \mathcal{Z}_l^i
\end{align}
which proves the satisfaction of (6e).
\begin{align}
&\parallel\sum_{i=1}^{M}\Psi^i_x{A^i}^le^i(t_k+M_k^i|t_k)\parallel \le \sum_{i=1}^{M}\parallel\Psi^i_x\parallel\parallel{A^i}\parallel^l\parallel e^i(t_k+M_k^i|t_k)\parallel\notag\\
&\le\sum_{i=1}^{I}\parallel\Psi_x^i\parallel\bar{w^i}\frac{\parallel A^i \parallel^{l}-\parallel A^i \parallel^{l+M_k^i}}{1-\parallel A^i \parallel} + \sum_{j=1}^{J}\parallel\Psi_x^j\parallel M_k^i\bar{w^j}\notag\\
&\le\epsilon(M_k^i+l)-\epsilon(l)\notag\\
&\sum_{i=1}^{M}(\Psi^i_x\hat{z}^i(t_{k+1}+l|t_{k+1})+\Psi^i_u\hat{u}^i(t_{k+1}+l|t_{k+1}))\notag\\
&=\sum_{i=1}^{M}(\Psi^i_xz^i(t_{k+1}+l|t_k) + \Psi^i_uu^i(t_{k+1}+l|t_k)) + \sum_{i=1}^{M}\Psi^i_x{A^i}^le^i(t_k+M_k^i|t_k)  \notag\\
&\le(1- \epsilon(M_k^i+l))\textbf{1}_p + \sum_{i=1}^{M}\Psi^i_x{A^i}^le^i(t_k+M_k^i|t_k)\notag\\
&\text{So we have}\notag\\
&\sum_{i=1}^{M}(\Psi^i_x\hat{z}^i(t_{k+1}+l|t_{k+1})+\Psi^i_u\hat{u}^i(t_{k+1}+l|t_{k+1}))\notag\\
&\le(1- \epsilon(M_k^i+l))\textbf{1}_p + (\epsilon(M_k^i+l)-\epsilon(l))\textbf{1}_p
\le\epsilon(l)\textbf{1}_p\notag
\end{align}
it prove the global constraints (15a).\\\\
For $l\in\mathbb{N}_{[N-M^i_k, N-1]}$, 
\begin{align}
\hat{z}^i(t_k+N|t_{k+1}) = z^i(t_k+N|t_k) + {A^i}^{N-M_k^i}e^i(t_k+M_k^i|t_k) \notag
\end{align}
Then we hold from triangle inequality that 
\begin{align}
\parallel\hat{z}^i(t_k+N|t_{k+1})\parallel_{P^i}
&\le \parallel z^i(t_k+N|t_k)\parallel_{P^i}
+ \parallel{A^i}^{N-M_k^i}e^i(t_k+M_k^i|t_k)\parallel_{P^i}\notag\\
&\le \parallel z^i(t_k+N|t_k)\parallel_{P^i}
+ \parallel e^i(t_k+N|t_k)\parallel_{P^i}\notag\\
&\begin{cases}
\le \parallel z^i(t_k+N|t_k)\parallel_{P^i} + \bar{w}^i\sqrt{\lambda_{max}(P^i)}\frac{1-\parallel A^i \parallel^N}{1-\parallel A^i \parallel} , \quad ||A^i||\neq 1\\
\le \parallel z^i(t_k+N|t_k)\parallel_{P^i} + \bar{w^i}\sqrt{\lambda_{max}(P^i)}N, \quad ||A^i|| =  1\\
\end{cases}\notag
\end{align}
\indent Considering the local condition  $\bar{w^i}\le \frac{(r^i - \varepsilon^i)(1-||A^i||)}{\sqrt{\lambda_{max}(P^i)}(1-||A^i||)^N} $ ($\bar{w^i}\le \frac{(r^i - \varepsilon^i)}{\sqrt{\lambda_{max}(P^i)}N}, when ||A^i|| = 1$) and $z^i(t_k+N|t_k) \in \mathcal{X}_\varepsilon^i$, we have $\hat{z}^i(t_k+N|t_{k+1}) \in \mathcal{X}_r^i$.\\ 
And the set $\mathcal{X}^i_r$ is maximal Robust Control Invariant Set(RCIS) under local state feedback control law : $u^i(z^i) = K^iz^i$, so we have 
\begin{align}
\hat{z}^i(t_{k+1}+l|t_{k+1}) &= (A^i+ B^iK^i)\hat{z}^i(t_{k+1}+l-1|t_{k+1})\notag\\
&\in \mathcal{X}_r^i \subseteq \mathcal{Z}_N^i \subseteq \mathcal{Z}_l^i,
\end{align}
which proves the satisfaction of (6e), from (7), we know $K^i\hat{z}^i \in \mathcal{U}^i$,which satisfies (6f)\\
from lemma 3, when $\hat{z}^i(t_{k+1}+l|t_{k+1}) \in\mathcal{X}_r^i$, we have
\begin{align}
\sum_{i=1}^{M}&(\Psi^i_x\hat{z}^i(t_{k+1}+l|t_{k+1})+\Psi^i_u\hat{u}^i(t_{k+1}+l|t_{k+1}))\notag\\
&= \sum_{i=1}^{M}(\Psi^i_x\hat{z}^i(t_{k+1}+l|t_{k+1})+\Psi^i_uK^i\hat{z}^i(t_{k+1}+l|t_{k+1}))\notag\\
&= \sum_{i=1}^{M}\Psi^i_N\hat{z}^i(t_{k+1}+l|t_{k+1})\notag\\
&\le (1- \epsilon(N)) \textbf{1}_p  \le(1- \epsilon(l)\textbf{1}_p \notag
\end{align}
Next we prove the satisfaction of (6c). With a proper choose of $P^i$ which obtained by (9), we have 
\begin{align}
&\parallel\hat{z}^i(t_{k+1}+N|t_{k+1})\parallel_{P^i}^2 - \parallel\hat{z}^i(t_{k+1}+N-1|t_{k+1})\parallel_{P^i}^2\notag\\
&= -\parallel\hat{z}^i(t_{k+1}+N-1|t_{k+1})\parallel_{Q^i + {K^i}^TR^iK^i}^2\notag\\
&\le -\parallel\hat{z}^i(t_{k+1}+N-1|t_{k+1})\parallel_{Q^i}^2\notag\\
&\le -\lambda_{min}(Q^i)\parallel\hat{z}^i(t_{k+1}+N-1|t_{k+1})\parallel^2\notag\\
&\le-\frac{\lambda_{min}(Q^i)}{\lambda_{max}(P^i)}\parallel\hat{z}^i(t_{k+1}+N-1|t_{k+1})\parallel_{P^i}^2\notag
\end{align}
from above inequality and (17), we have
\begin{align}
\parallel\hat{z}^i(t_{k+1}+N|t_{k+1})\parallel_{P^i}^2 &\le(1-\frac{\lambda_{min}(Q^i)}{\lambda_{max}(P^i)})\parallel\hat{z}^i(t_{k+1}+N-1|t_{k+1})\parallel_{P^i}^2\notag\\
&\le(1-\frac{\lambda_{min}(Q^i)}{\lambda_{max}(P^i)}){r^i}^2 \le {\varepsilon^i}^2
\end{align}
So we have $\hat{z}^i(t_{k+1}+N|t_{k+1}) \in \mathcal{X}_\varepsilon^i$ which satisfies with (6c), from lemma 3 we get the following inequality \\
\begin{align}
&\sum_{i=1}^{M}\Psi^i_N\hat{z}^i(t_{k+1}+N|t_{k+1})	\le(1- \epsilon(N))\textbf{1}_p\notag
\end{align}
Hence we complete the proof of $(6b)-(6f)$ and $(15)$ in $\mathbb{P}_\epsilon(x(t_{k+1}))$.
\end{proof}
\end{document}